\documentclass[letterpaper, 10 pt, conference]{ieeeconf}  

\IEEEoverridecommandlockouts                              
\usepackage{amsmath, amssymb, amsfonts, mathtools}
\usepackage{graphicx}
\usepackage{algorithm}
\usepackage{algorithmic}
\usepackage{bm}
\usepackage{tikz} 
\usepackage{pgfplots} 
\pgfplotsset{compat=1.17} 
\usepackage{ulem}
\usetikzlibrary{shapes.geometric, arrows} 
\usepackage{comment}
\usepackage{booktabs}
\usepackage{adjustbox}
\usepackage{array} 
\usepackage[hidelinks]{hyperref}

\expandafter\def\expandafter\normalsize\expandafter{%
    \normalsize%
    \setlength\abovedisplayskip{2pt}%
    \setlength\belowdisplayskip{2pt}%
    \setlength\abovedisplayshortskip{4pt}%
    \setlength\belowdisplayshortskip{0pt}%
}

\newtheorem{theorem}{Theorem}
\newtheorem{lemma}{Lemma}
\newtheorem{definition}{Definition}

\newtheorem{proposition}{Proposition}

\newtheorem{remark}{Remark}

\definecolor{green1}{rgb}{0.2,0.7,0.5}

\definecolor{green1}{rgb}{0.2,0.7,0.5}

\title{ \LARGE \bf A Soft Inducement Framework for Incentive-Aided Steering of No-Regret Players}
\author{ Asrin Efe Yorulmaz,  Raj Kiriti Velicheti, Melih Bastopcu, and  Tamer~Ba{\c s}ar 
\thanks{Research of AEY, RKV, and TB was supported in part by the US Army Research Office (ARO) Grant W911NF-24-1-0085 and in part by the NSF Grant ECCS 23-49418. Research of MB was supported in part by Tubitak Bilgem EDGE-4-IoT and Tubitak 2232-B Fellowship (Project No:124C533). Asrin Efe Yorulmaz, Raj Kiriti Velicheti, and Tamer~Ba{\c s}ar are with the Coordinated Science Laboratory at the University of Illinois Urbana-Champaign, Urbana, IL, USA-61801. Melih Bastopcu is with the Department of Electrical and Electronics Engineering, Bilkent University, Ankara, Turkey, 06800. (Emails: \texttt{\{ay20,rkv4,basar1\}@illinois.edu}, \texttt{bastopcu@bilkent.edu.tr})}%
}

\begin{document}

\maketitle

\maketitle

\begin{abstract}
In this work, we investigate a steering problem in a mediator-augmented two-player normal-form game, where the mediator aims to guide players toward a specific action profile through information and incentive design. We first characterize the games for which successful steering is possible. Moreover, we establish that steering players to any desired action profile is not always achievable with information design alone, nor when accompanied with sublinear payment schemes. Consequently, we derive a lower bound on the constant payments required per round to achieve this goal. To address these limitations incurred with information design, we introduce an augmented approach that involves a one-shot information design phase before the start of the repeated game, transforming the prior interaction into a Stackelberg game. Finally, we theoretically demonstrate that this approach improves the convergence rate of players' action profiles to the target point by a constant factor with high probability, and support it with empirical results.
\end{abstract}
\section{Introduction}
In strategic interactions, information plays a crucial role in shaping the decisions of agents. A mediator can influence the outcome of a game by controlling the information available to the players \cite{bacsar2024inducement}. This control is particularly valuable in multi-agent settings, where individual decisions have system-wide effects \cite{babichenko2021multi}. We motivate the integration of information design into a steering problem by highlighting its applications across economics, governance, and digital markets. In economic contexts, persuasion mechanisms are widely employed to influence consumer behavior alongside monetary incentives \cite{mccloskey1995one}. Governments use public information campaigns to encourage participation in social programs \cite{egorov2020persuasion}, advertisers strategically release information to shape consumer preferences \cite{johnson2006simple}, and online platforms curate content to influence user engagement \cite{ke2022information}. Similarly, financial markets and regulatory bodies use signaling mechanisms to steer investor expectations \cite{goldstein2016bayesian}, affecting market conditions \cite{goldstein2018stress}.

While long-term (asymptotic) effects are often the focus of traditional analysis, many real-world scenarios demand a stronger emphasis on transients. For instance, in emergency situations such as disaster management \cite{bosmans2022systematic} or financial crises \cite{blinder2015financial}, ensuring short-term adaptation is crucial for mitigating immediate risks. These considerations motivate our results on improving transient regret, ensuring that strategic decisions align with objectives in both the short run and the long run.

Our work focuses on an information-aided incentive design problem in a two-player normal-form “investment game”, which is repeated over \( T \) rounds. Also, due to the nature of the information design problem, the games we consider are inherently Bayesian Normal Form Games (BNFGs) \cite{harsanyi1967}. The players, modeled as no-regret learners employing EXP3.P algorithms \cite{bubeck2012regret}, receive public signals about the state of the world from the mediator and accordingly choose actions, as the state of the world is not observable to them.

The mediator’s objective is to steer the players toward a specific strategy profile in an empirical sense. However, we demonstrate that achieving this objective using information design alone is not always feasible. Furthermore, we show that this goal cannot be accomplished with sublinear payments for all cases. The results on the feasibility of successful steering are provided in Table~\ref{Table_1}. 

To address these limitations we propose an augmented method that incorporates a round of information designs prior to the repeated game, which can be modeled as a Stackelberg game \cite{von1934marktform}. This provides better initial conditions for the players, thereby improving the convergence rate of their directness gap to
\( 
\tilde{\mathcal{O}}\!\left(\frac{\sqrt{K}}{\kappa\sqrt{T}}\left(4\sqrt{\ln(1/\pi^*)} + 2\sqrt{\ln(K/\delta)}\right)\right)
\)   
for each signal instance, where \(\kappa\) is the minimum deviation cost from the best hindsight action, and \( \pi^* \) is the probability of choosing the best hindsight action, compared to the usual  
\(
\tilde{\mathcal{O}}\left(\frac{\sqrt{K}}{\kappa\sqrt{T}}\left(4\sqrt{\ln(K)} + 2\sqrt{\ln(K/\delta)}\right)\right)
\)   
rate, which holds with probability of at least \( 1-\delta \), where \( K \) represents the number of available actions and \( T \) denotes the number of time steps the game has progressed \cite{bubeck2012regret}.\footnote{We use \(\tilde{\mathcal{O}}(\cdot)\) to denote the leading term in \(T\) and its constants.} Thus, the improvement for each signal instance improves the overall convergence bound. 

Our work aligns with research on steering rational agents through incentive design \cite{mguni2019coordinating, liu2022inducing}. The closest work to our setting is \cite{zhang2023steering}, which incorporates both incentive and information design. However, their focus is on how to guide no-regret agents toward a Nash equilibrium and its variants, whereas our goal is a more general notion of action profiles. Additionally, we optimize for mediator's signaling strategy, which we demonstrate to be non-trivial. Table~\ref{Table_1} summarizes our main contributions and contrasts them with key results from a closely related work of \cite{zhang2023steering}.\footnote{In Table~\ref{Table_1}, a $\checkmark$ denotes that the given action profile can be reached by the corresponding method, and a $\times$ indicates existence of a counterexample.}

Our work also relates to the strategic communication problem, originating from \cite{crawford1982strategic} and formalized as Bayesian persuasion in \cite{kamenica2011bayesian}. Furthermore, we leverage extensions of correlated equilibrium to Bayesian games from \cite{bergemann2016bayes}. Lastly, our approach relates to multi-armed bandits in adversarial environments, as repeated normal-form games can be modeled in this framework. Thus, we leverage results from \cite{hartline2015no} to analyze no-regret learning algorithms in Bayesian games.

\begin{table}[h]
\centering
\caption{Comparison of Game Types and Properties}\label{Table_1}

\begin{adjustbox}{max width=\textwidth}
\begin{tabular}{|>{\centering\arraybackslash}p{1.22cm}|>{\centering\arraybackslash}p{1.25cm}|>{\centering\arraybackslash}p{1.35cm}|>{\centering\arraybackslash}p{1.25cm}|>{\centering\arraybackslash}p{1.25cm}|}
\hline
\shortstack{\textbf{Property}\\\textbf{}\\\textbf{}\\\textbf{}  \\\textbf{}  } & \shortstack{\textbf{}\\\textbf{\!\!BNFG with}\\\textbf{\!\!Strictly}\\\textbf{\!\!Dominating }\\\textbf{\!\!Strategies}} & \shortstack{\textbf{\!\!\!Perfect}\\\textbf{\!\!\!Information}\\\textbf{\!\!\!Normal Form}\\\textbf{\!\!\!Games}} & \shortstack{\textbf{\!BNFG w/o}\\\textbf{\!Strictly}\\\textbf{\!Dominating}\\\textbf{\!Strategies}} & \shortstack{\textbf{\!BNFG w/o}\\\textbf{\!Strictly}\\\textbf{\!Dominating}\\\textbf{\!Strategies}} \\ 
\hline
\raisebox{-1ex}{Targeted} Points & \raisebox{-0.75ex}{Dominant} \raisebox{-0.75ex}{Strategy Eq.} & \raisebox{-0.75ex}{Pure}\raisebox{-0.75ex}{ Nash} \raisebox{-1ex}{Equilibrium} & \raisebox{-1.80ex}{Bayes CE} & \raisebox{-0.75ex}{Mediator} \raisebox{-0.75ex}{Decided Pt.} \\
\hline
Information (Advice) & \raisebox{-1.25ex}{\checkmark \hspace{2pt}(Lem.\hspace{2pt}\ref{lemma_dom})}  & \raisebox{-1.25ex}{-}  & \raisebox{-1.25ex}{$\times$\hspace{2pt}(Thm.\hspace{2pt}\ref{Lemma3})}  & \raisebox{-1.25ex}{$\times$\hspace{2pt}(Thm.\hspace{2pt}\ref{Lemma3})}  \\
\hline
Sublinear Payments & \raisebox{-1.25ex}{\checkmark \hspace{2pt}(Lem.\hspace{2pt}\ref{lemma_dom})} & \raisebox{-1.25ex}{\checkmark\hspace{2pt}\cite{zhang2023steering}} & \raisebox{-1.25ex}{$\times$\hspace{2pt}\cite{zhang2023steering}} & 
\raisebox{-1.25ex}{$\times$\hspace{2pt}(Thm.\hspace{2pt}\ref{Lemma4})}  \\
\hline
Sublinear Payments \& Advice & \raisebox{-1.25ex}{\checkmark \hspace{2pt}(Lem.\hspace{2pt}\ref{lemma_dom})} & \raisebox{-1.25ex}{-} & \raisebox{-1.25ex}{\checkmark\hspace{2pt}\cite{zhang2023steering}} & \raisebox{-1.25ex}{$\times$\hspace{2pt}(Thm.\hspace{2pt}\ref{Lemma4})} \\
\hline
Linear Payments & \raisebox{-1.25ex}{\checkmark\hspace{2pt}(Lem.\hspace{2pt}\ref{lemma_dom})} & \raisebox{-1.25ex}{\checkmark\hspace{2pt}\cite{zhang2023steering}} & \raisebox{-1.25ex}{\checkmark\hspace{2pt}\cite{zhang2023steering}} & \raisebox{-1.25ex}{\checkmark\hspace{2pt}\cite{zhang2023steering}} \\
\hline
\end{tabular}
\end{adjustbox}

\end{table} 

\section{Problem Formulation}

We consider a steering problem involving a mediator in a game with two players, each making investment decisions under uncertainty. This uncertainty is captured by the state of the world denoted by \( \theta \in  \Theta= \{G, B\}  \), which is not directly observable by the players but is known and leveraged by the mediator. The prior probability of the good state ($G$) is given by \( \psi(G) = \psi \), and of the bad state ($B$) with \( \psi(B) = 1 - \psi \), where \( 0 < \psi < 1 \), and is known to all players. 

Each player $i \in \mathcal{I} = \{1,2\}$ chooses an action $a_i \!\in\! \mathcal{A}\! = \!\{I, N\}$, where $I$ denotes \textit{invest} and $N$ denotes \textit{not invest}. Furthermore, we denote the other player's actions by $a_{-i} \in \mathcal{A}$. Each player aims to maximize its received utility in each repetition of the game by trying to leverage the information provided by the mediator. To capture inter-player externalities, we introduce feature vectors $ \mathbf{f}_1, \mathbf{f}_2 \in \mathbb{R}^d \!$, representing characteristics of players 1 and 2, respectively. The externality parameter \( z \) is modeled as a function of the alignment between these feature vectors, \( z = \phi\left( \langle \mathbf{f}_1, \mathbf{f}_2 \rangle \right),\) where \( \langle \cdot, \cdot \rangle \) denotes the inner product. The function \( \phi: \mathbb{R} \to \mathbb{R}_+ \) is monotonically increasing, reflecting the intuition that greater alignment between the players' features leads to stronger externalities. Furthermore, we assume symmetry among the players; thus, players share identical preferences over outcomes and are equally affected by the externality parameter \( z \).

In contrast, the mediator aims to steer players into the point of her desire. In our framework, the mediator, who knows the game structure and players being no-regret learners, provides monetary incentives along with public signaling regarding the state of the world. We utilize the standard definition of regret formalized as follows.
\begin{definition}
Let $\mathcal{A}$ be the set of possible actions available to a player. Suppose that the player selects an action $x_t \in \mathcal{A}$ at each time step $t \in \{1, \cdots, T\}$ and receives a corresponding reward $u_t(x_t)$. Then, the \textit{external regret} is defined as:
\begin{align} \nonumber\\[-19pt]
    R(T) = \max_{x^* \in \mathcal{A}} \sum_{t=1}^{T} u_t(x^*) - \sum_{t=1}^{T} u_t(x_t).  \label{regret}  
\end{align}
\end{definition}

Formally, the mediator first commits to a stationary signaling policy \( \pi(s|\theta) \), which specifies the probability of sending signal \(s \in \{g, b\} = S\) given the underlying state \(\theta\). The signaling strategy is parameterized as, $ 
\pi(g|G) =1-\pi(b|G)= \alpha,$ $\pi(g|B) = 1-\pi(b|B)= \beta,$ where \(0 \leq \alpha, \beta \leq 1\). Furthermore, we assume these parameters stays constant throughout the repeated game. Additionally, the mediator can select a payment function
$\nu_i: \mathcal{A} \times \mathcal{A} \to [0,P]$ for each player \( i \), where each \( \nu_i \) is continuous in the player's action. The modified utility for player $i$ becomes \( 
v_i^{(t)}(a^t_i, a_{-i}^t, \theta_t) = u_i(a^t_i, a_{-i}^t, \theta_t) + \nu _i(a^t_i, a_{-i}^t).
\) Furthermore, we introduce the notion of the directness gap, which can be defined over the deviations between the joint actions of the players, \( x^{(t)}\), and the target profile of the mediator  \( d \in \mathcal{A} \times \mathcal{A}\) throughout the game, as:
\begin{align}\notag\\[-18pt] 
\delta(T) = \frac{1}{T} \sum_{t=1}^T \mathbf{1}\{x^{(t)} \neq d\}.
\end{align}

We now introduce the steering problem. In general, the steering problem asks whether and under what conditions players can be guided to a specific action profile. The mediator's objectives are twofold. First, the time-averaged payments must be minimized. Second, the players' actions should become indistinguishable from the target equilibrium \( d \in \mathcal{A} \times \mathcal{A} \), meaning that the directness gap converges to 0. Thus, we investigate how information design can be incorporated to improve the steering capabilities of the mediator in games that involve information asymmetry. To denote how the repeated game proceeds, we define the decision rule $\sigma_{i,t}(a_i|s)$ as a time-dependent variable that stands for player $i$ choosing action $a_i$ given the signal $s$ at time step $t$. Then, denoting the history of the actions until the timestep $t$ by $h_t$, the game proceeds at each time-step as follows:
\begin{enumerate}
    \item Mediator commits to a stationary signaling policy, $\pi(\cdot|\theta)$, and a stationary incentive mechanism, $\nu_i$. 
    \item At time $t$, the state, \(\theta_t\), is realized.
    \item Mediator samples a public signal \(s_t \sim \pi(\cdot | \theta_t)\) from the signaling policy \(\pi\).
    \item The players observe the signal \(s_t\), and each player \( i \) samples its action \( a_i^t \sim \sigma_{i,t}(\cdot | s_t,h_t) \).
    \item Each player receives a reward based on its actions and the realized state, \( 
    v_i^{(t)} \!=\! u_i(a^t_i, a_{-i}^t, \theta_t) \!+ \!\nu_i(a^t_i, a_{-i}^t).
    \)
    \item Players update their strategies, $\sigma_{i,t}(\cdot | s_{t+1},\!h_{t+1})$, according to $v^{(t)}_i$. Steps 2–6 are repeated until $t=T$.
\end{enumerate}

As, the mediator chooses its policy based on the underlying static game, we introduce a suitable notion of equilibrium and solution set. The payoff matrix of the focused game is given by:  
\begin{align*}
\begin{tabular}{c|cc}
\centering
 & Player 2: $I$ & Player 2: $N$ \\ \hline
Player 1: $I$ & $(z + y_\theta, z + y_\theta)$ & $(z, 0)$ \\
Player 1: $N$ & $(0, z)$ & $(0, 0)$ \\
\end{tabular}
\end{align*}
Here, $y_\theta$ denotes either $y_G$ or $y_B$ based on the realized state $\theta \in \{G,B\}$; $y_G$ and $y_B$ are parameters representing the additional payoffs in good and bad states, respectively, with $y_B < 0 < y_G$. Then, the stationary  joint action matrix formed according to a provided signal \( s_j \) is given as:
\begin{align*}
\centering
\begin{tabular}{c|cc}
 & Player 2: $I$ & Player 2: $N$ \\ \hline
Player 1: $I$ & $\gamma_j$ & $\alpha_j - \gamma_j$ \\
Player 1: $N$ & $\alpha_j - \gamma_j$ & $1 - 2\alpha_j + \gamma_j$ \\
\end{tabular}
\end{align*}

Here, $\alpha_j$ is the probability of a player playing action $I$ given signal $s_j$ and $\gamma_j$ is the probability that both players play action $I$ given signal $s_j$. These probabilities satisfy $0 \leq \gamma_j \leq \alpha_j \leq 1$ and $1-2\alpha_j+\gamma_j \geq 0$. Letting \(\sigma(a \mid s)\), or \(\sigma(a_i,a_{-i} \mid s)\) denote the stationary probability of joint action \(a\) given signal \(s\), the expected utility for player \(i\) is then defined as: 
\begin{align}\label{eq:expected_utility_final}
    \mathbb{E}[u_i] &= \sum_{\theta} \psi(\theta) \sum_{s} \pi(s|\theta)  
    \sum_{a_{i}}\sum_{a_{-i}} \sigma(a_i, a_{-i}|s)u_i(a_i, a_{-i}, \theta) \nonumber\\[-6pt]
    &= \psi \pi(g|G) \left[ \gamma_g (z + y_G) + (\alpha_g - \gamma_g) z \right]\nonumber \\
    &\quad + \psi \pi(b|G) \left[ \gamma_b (z + y_G) + (\alpha_b - \gamma_b) z \right] \nonumber\\
    &\quad + (1 - \psi) \pi(g|B) \left[ \gamma_g (z + y_B) + (\alpha_g - \gamma_g) z \right]\nonumber \\
    &\quad + (1 - \psi) \pi(b|B) \left[ \gamma_b (z + y_B) + (\alpha_b - \gamma_b) z \right].
\end{align}
Thus, substituting the signaling probabilities, \(\pi(g|G) = \alpha\) and \(\pi(g|B) = \beta\), we obtain:
\begin{align}
    \mathbb{E}[u_i] &= \psi \alpha ( \gamma_g y_G + \alpha_g z ) + \psi (1 - \alpha) (\gamma_b y_G + \alpha_b z)  \\
    & +\! (1\! - \!\psi) \beta ( \gamma_g y_B  \!+\! \alpha_g z )\! + \! (1 \!-\! \psi) (1 \!-\! \beta) (\gamma_b y_B \!+ \!\alpha_b z).\nonumber
\end{align}

Following this, a natural notion of equilibrium in games with publicly observed signals is \textit{Bayes Correlated Equilibrium (BCE)} \cite{bergemann2016bayes}, defined formally as follows.  
\begin{definition}
A strategy profile \(\sigma(a | s)\) is said to belong to the BCE set if no player has an incentive to deviate for any signal \(s\) and any alternative action \(a_i' \in \{I, N\}\). This condition, known as the BCE compliance, is given by:
\begin{align}
\label{eq:bce_constraint}
\sum_{\theta}\sum_{a_{-i}} \psi(\theta) \pi(s|\theta) &\big(\sigma(a_i, a_{-i}|s) u_i(a_i, a_{-i}, \theta) \nonumber \\[-7pt] 
&- \sigma(a_i', a_{-i}|s) u_i(a_i', a_{-i}, \theta)\big) \geq 0.
\end{align}
\end{definition}

Additionally, the variables must satisfy the following feasibility conditions as mentioned above:
\begin{equation}
    0 \leq \gamma_j \leq \alpha_j \leq 1, \quad 1-2\alpha_j+\gamma_j \geq 0, \quad \forall j \in \{g, b\}.
\end{equation}

We next address steering no-regret players toward a target strategy profile.

\section{Steering No-regret Players Toward the Target Strategy Point}

The problem of steering no-regret players toward a specific strategy point in our setting can be reduced, without loss of generality, to the problem of guiding players to \( (I,I) \). This is because the role of \( (I,I) \) ranges from being strictly dominant to being dominated as we vary the utilities. To achieve this, we formalize the definitions of no-regret learning dynamics and Bayes-CCE (BCCE) set, and provide a proof of convergence of no-regret players to such an equilibrium. This approach builds upon analogous results presented in \cite{hartline2015no}, where the proof was provided based on population-based interpretation of Bayesian games. Formally, we define the no-regret property as follows: 
\begin{definition}
An algorithm is said to be a \textit{no-regret algorithm} if its external regret, as defined in \eqref{regret}, grows sublinearly with respect to the number of time steps \( T \). Formally, we have:
\begin{align}
    \!\!\!\!\!\frac{R(T)}{T} \!=\! \frac{1}{T} \!\left( \!\max_{x^* \in \mathcal{A}} \sum_{t=1}^{T} u_t(x^*) \!-\! \!\sum_{t=1}^{T} \!u_t(x_t) \!\right) \!\to \!0  \text{ as } T \!\!\to \!\infty.\!
\end{align}
\end{definition}
Furthermore, we define a BCCE set as follows. 
\begin{definition}
A strategy profile \(\sigma(a| s)\) is said to belong to the BCCE set if, for each player \( i \), and any alternative action \( a_i' \in \mathcal{A}_i \), the following condition holds:
\begin{align}
\mathbb{E}_{\psi, \pi, \sigma} \Big[  u_i(a_i, a_{-i}, \theta) \Big] 
 \geq 
\mathbb{E}_{\psi, \pi, \sigma} \Big[ u_i(a_i', \!a_{-i}, \theta) \Big].
\end{align}
\end{definition}
Then, for the convergence of no-regret players in repeated BNFGs, we present the following lemma.
\begin{lemma}\label{lem:bcc} Under separate no-regret learning dynamics for each signal instance, the joint empirical distribution
\begin{align}\label{eqn:emp_dist}
 \!\!\! \!  D_T(\theta,s,a_1,a_2) \!=\! \!\frac{1}{T}\!\sum_{t=1}^{T} \mathbf{1}\{\theta_t\!=\!\theta, s_t\!=\!s, a_1^t\!=\!a_1, a_2^t\!=\!a_2\} \!\!\!  
\end{align}
converges almost surely to $ 
D(\theta,s,a_1,a_2)\!=\!\psi(\theta)\,\pi(s\mid\theta)$ $\sigma(a_1,a_2\mid s)$. Moreover, \(D(\theta,s,a_1,a_2)\) satisfies the BCCE conditions.
\end{lemma}
\begin{proof}
\!\!We take the joint empirical distribution as in (\ref{eqn:emp_dist}). 
\begin{align*}
 \!\!\! \!  D_T(\theta,s,a_1,a_2) \!=\! \!\frac{1}{T}\!\sum_{t=1}^{T} \mathbf{1}\{\theta_t\!=\!\theta, s_t\!=\!s, a_1^t\!=\!a_1, a_2^t\!=\!a_2\}. \!\!\!  
\end{align*}
Assuming that each empirical frequency converges to some stationary probability distribution, which we will show to be the case in later steps, we can write down the following expression:
\[
D_T(\theta,s,a_1,a_2) \xrightarrow{T\to\infty} P(\theta,s,a_1,a_2).
\]
Since, the players' actions depend only on the received signals, we can express the joint probability distribution as:
\[
P(\theta,s,a_1,a_2) = P(\theta,s) P(a_1,a_2 \mid s).
\]
Since the states \(\{\theta_t\}\) are i.i.d.\ with distribution \(\psi(\theta)\), and given \(\theta_t\), the signal \(s_t\) is drawn according to \(\pi(s\!\mid\!\theta_t)\), the pairs \((\theta_t,s_t)\) are i.i.d.. Hence, by the Strong Law of Large Numbers (SLLN),
\begin{align*}\\[-18pt]
 \frac{1}{T}\sum_{t=1}^{T} \mathbf{1}\{\theta_t=\theta,\; s_t=s\} \xrightarrow{\text{a.s.}} \psi(\theta)\,\pi(s\mid\theta)\\[-18pt]   
\end{align*}
for every \((\theta,s)\in \Theta\times S\). To show the convergence of the action profiles, fix a signal \(s\in S\) with \(P(s)>0\) and let \( 
T_s=\{t\le T: s_t=s\}.\) For each fixed action profile \((a_1,a_2)\in \mathcal{A}_1\times \mathcal{A}_2\), define the empirical frequency of actions for \(t\in T_s\) as, \( 
X_t = \mathbf{1}\{a_t^1=a_1,\; a_t^2=a_2\}.
\)
Then, the time-averaged frequency over rounds when \(s_t=s\) is given by
\begin{align*}\\[-18pt]
  \overline{X}_{|T_s|} = \frac{1}{|T_s|}\sum_{t\in T_s} X_t.\\[-20pt]  
\end{align*}
Since \(P(s)>0\), we have \(|T_s|\to\infty\) as \(T\to\infty\) almost surely. Moreover, it is well known by  \cite{cesa2006prediction} that, when all players use no-regret algorithms, the empirical frequencies of players' actions almost surely converge to the CCE of the static game. Also, it can be observed that, due to the SLLN, the i.i.d. and stationary nature of the provided signals and state transition probabilities, as \( T \to \infty \), for each given signal instance, an ``expected" or ``static" game is formed, where the payoffs of the players are sampled from an i.i.d. distribution accordingly. Thus, no-regret algorithms guarantee the convergence in the new ``static'' game. Hence, we obtain
\begin{align*}
 \overline{X}_{|T_s|} \xrightarrow{\text{a.s.}} \sigma(a_1,a_2\mid  s),   
\end{align*}
where \(\sigma(\cdot\mid s)\) is the limiting distribution over \(\mathcal{A}_1\times \mathcal{A}_2\) conditioned on the signal \(s\), subject to the CCE set of the stationary game. 
Therefore, it follows that:
\begin{align*}
 D_T(\theta,s,a_1,a_2) \!\!\xrightarrow{\text{a.s.}}\!\! D(\theta,s,a_1,a_2) \!\!=\!\! \psi(\theta)\,\pi(s\!\mid\theta)\,\sigma(a_1,a_2\!\mid\! s).   
\end{align*}
Since the state, signal, and action spaces are finite, and the payoff functions \(u_i(a_i,a_{-i},\theta)\) are bounded, we define, for each player \(i\), the mapping given the signal \(s\) as:
\begin{align*}
   \!F_i(D_T(\theta,\!s,\!a_i,\!a_{-i})\!)\!\!=\!\! \sum_{\theta\in \Theta} \sum_{a\in A} \!D_T(\theta,\!s,\!a_i,\!a_{-i})u_i(a_i,\!a_{-i},\!\theta).
\end{align*}
Thus, it is easy to see that the function \(D_T\mapsto F_i(D_T)\) is continuous. Given that, $a_i$ is the action profile of player $i$, the no-regret property guarantees that for every alternative fixed action \(a_i'\in \mathcal{A}_i\), the empirical distributions satisfy an approximate no-deviation inequality:
\begin{align*}
F_i(D_T(\theta,s,a_i,a_{-i})) \geq F_i(D_T(\theta,s,a'_i,a_{-i})) - \epsilon_T,
\end{align*}
\normalsize
with \(\epsilon_T\to 0\) as \(T\to\infty\). By the Continuous Mapping Theorem \cite{shao2003mathematical}, taking the limit as \(T\to\infty\) yields
\begin{align*}
  F_i(D(\theta,s,a_i,a_{-i})) \geq F_i(D(\theta,s,a'_i,a_{-i})),  
\end{align*}
for all \(a_i,a_i'\in \mathcal{A}_i\) and for each player \(i\). Then, summing both sides over the all possible signals, this expression becomes precisely the BCCE condition:
\begin{align*}
    \mathbb{E}_D\Bigl[u_i(a_i,a_{-i},\theta)\Bigr] \ge \mathbb{E}_D\Bigl[u_i(a_i',a_{-i},\theta)\Bigr],
\end{align*}
which concludes the proof.
\end{proof}

Having established convergence to the BCCE set under separate no-regret algorithms conditioned on the given signal instance, we now present the EXP3.P in Algorithm \ref{alg:exp3p}. 

The EXP3.P algorithm operates over an adversarial multi-armed bandit setting with $K$ arms, different actions, for $T$ rounds. It maintains a probability vector $p_t\in\Delta_K$ and cumulative gain estimates $\{\widehat G_{i,t}\}_{i=1}^K$. Initially, $p_1$ is uniform over all arms and $\widehat G_{i,0}=0$ for each $i$. At round $t$, the learner samples an arm $I_t\sim p_t$ and observes a gain $g_{I_t,t}\in[0,1]$. To correct for partial feedback and to obtain high-probability concentration, and limit the variance of gain estimates EXP3.P forms the importance-weighted estimate \( 
\widehat g_{i,t}
\), provided in Line 5 of Algorithm \ref{alg:exp3p}
with bias term $\beta\in[0,1]$, and updates $\widehat G_{i,t}=\widehat G_{i,t-1}+\widehat g_{i,t}$. The next distribution is obtained by exponentiating the scaled cumulative estimates and mixing with a uniform exploration floor
\( 
p_{t+1}(i)
\), given in Line 8 of Algorithm \ref{alg:exp3p},
where $\eta>0$ is the learning rate and $\gamma\in[0,1]$ controls minimum exploration. By intertwining the bias $\beta$ and the exploration floor $\gamma$, EXP3.P attains a regret bound of order $\mathcal{O}\bigl(\sqrt{T\,\ln(K/\delta)}\bigr)$ with probability at least $1-\delta$ \cite{bubeck2012regret}. In the following theorem, we relate each instance of the regret algorithm given $s$, to the overall regret.

\begin{algorithm}[H]
\caption{EXP3.P}
\label{alg:exp3p}
\begin{algorithmic}[1]
\REQUIRE Learning rate $\eta > 0$, parameters $\gamma, \beta,g_{i,t} \in [0,1]$, number of arms $K$
\STATE Initialize $p_1(i) \gets \frac{1}{K}$, and $\widehat{G}_{i,0} \gets 0$ for all $i \in \{1,\dots,K\}$
\FOR{$t = 1$ to $n$}
    \STATE Sample arm $I_t \sim p_t$
    \FOR{each arm $i = 1$ to $K$}
        \STATE Compute estimated gain:
        \( 
        \widehat{g}_{i,t} \gets \frac{g_{i,t}  \mathbf{1}\{I_t = i\} + \beta}{p_t(i)}
        \) 
       \STATE Update cumulative estimated gain:
        \( 
        \!\widehat{G}_{i,t} \!\!\gets\! \widehat{G}_{i,t\!-\!1} \!+\! \widehat{g}_{i,t}
        \) 
    \ENDFOR
    \STATE Update probabilities for next round:
    \[
    p_{t+1}(i) \gets (1 - \gamma) \cdot \frac{\exp(\eta \widehat{G}_{i,t})}{\sum_{k=1}^K \exp(\eta \widehat{G}_{k,t})} + \frac{\gamma}{K}
    \]
\ENDFOR
\end{algorithmic}
\end{algorithm}

\begin{theorem}
\label{thm:two-signal-exp3p-utilities}
Define the overall-regret across the signal instances by
\[
R_{ovr}(T)
:=
\sum_{s\in S}
\max_{a\in\mathcal{A}}
\sum_{t:\,s_t=s}\!\Big(u_{i,t}(a,a^t_{-i},\theta)-u_{i,t}(a^t_{i},a^t_{-i},\theta)\Big).
\]
Assume the following single-instance high-probability bound for \textsc{EXP3.P}:
there exists a constant $C>0$ such that for every horizon $T$, number of actions $|K|$ and confidence $\delta\in(0,1)$,
with probability of at least $1-\delta$, \( R(T) \;\le\;\tilde{R}(T;\delta) \;:=\; C\sqrt{K\,T\,\ln\!\Big(\tfrac{K}{\delta}\Big)}. \) 
Then, with probability of at least $1-\delta$,
\( 
R_{ovr}(T) \le \sqrt{2}\;\tilde{R}\Big(T;\tfrac{\delta}{2}\Big)
.
\) 
\end{theorem}

\begin{proof}
Let $n_s := \big|\{t\le T:\, s_t=s\}\big|$ denote the number of rounds in which signal $s$ appears, so that
$n_{s^{(1)}}+n_{s^{(2)}}=T$. Apply the single-instance high-probability bound to this subsequence
with confidence parameter $\delta_s>0$ to obtain the event
\( 
\mathcal{E}_s :=
\left\{
R(n_s)
\!\le \!\tilde{R}(n_s;\delta_s)
\right\},
\) 
which satisfies $\mathbb{P}(\mathcal{E}_s)\ge 1-\delta_s$. Choose $\delta_s=\delta/2$ for both signals and
invoke the union bound to get
\(
\mathbb{P}(\mathcal{E}_{s^{(1)}}\cap \mathcal{E}_{s^{(2)}})
\;\ge\; 1-\sum_{s\in S}\delta_s \;=\; 1-\delta.
\) 
On the intersection event $\mathcal{E}_{s^{(1)}}\cap \mathcal{E}_{s^{(2)}}$, we can sum the two bounds:
\[
R_{ovr}(T)
\;=\;
\sum_{s\in S}
R(n_s)
\;\le\;
\sum_{s\in S} \tilde{R}(n_s;\tfrac{\delta}{2}).
\]
By the assumed form of $\tilde{R}(\cdot\,;\cdot)$,
\[
\sum_{s\in S} \tilde{R}(n_s;\tfrac{\delta}{2})
\;=\;
C\sqrt{K\,\ln\!\Big(\tfrac{2K}{\delta}\Big)}\;\sum_{s\in S}\sqrt{n_s}.
\]
Finally, apply Jensen's inequality for the concave map:
\[
\sum_{s\in S}\sqrt{n_s}
\;\le\;
\sqrt{|S|\sum_{s\in S} n_s}
\;=\; \sqrt{2\,T}.
\]
Combining the last two displays yields, on $\mathcal{E}_{s^{(1)}}\cap \mathcal{E}_{s^{(2)}}$,
\[
R_{ovr}(T)
\le
C\sqrt{K\ln\!\Big(\tfrac{2K}{\delta}\Big)}\;\sqrt{2T}
=
\sqrt{2}\tilde{R}\Big(T;\tfrac{\delta}{2}\Big),
\]
which establishes the stated bound with probability of at least $1-\delta$.
\end{proof}

Upon reflecting on how to steer no-regret players toward specific BCCE points, it becomes evident that this occurs when the BCCE set is a singleton containing only the mediator's target action profile. From this, we identify two possible ways to achieve steering. The first case arises when each player has a strictly dominant action across all states. The second approach involves designing mechanisms to ensure that the BCCE set contains only a single point. Given these observations, next we formalize the first case.

\begin{lemma} \label{lemma_dom}
Let each player \( i \in \mathcal{I} \) have a strictly dominant action \( a_i^* \in \mathcal{A}_i \) in every state \( \theta \in \Theta \) and for every signal \( s \in S \) in the Bayesian game denoted by \( (\mathcal{I}, \mathcal{A}, \Theta, S,u) \). Then, the BCCE of the game is unique.
\end{lemma}

\begin{proof}
Since \( a_i^* \) is a strictly dominant action for player \( i \), it satisfies, \(
u_i(a_i^*, a_{-i}, \theta, s) > u_i(a_i, a_{-i}, \theta, s)
\), for every \( a_i \in \mathcal{A}_i \setminus\{a_i^*\} \), for all \( a_{-i} \in \mathcal{A}_{-i} \), \( \theta \in \Theta \), and \( s \in S \). In a BCCE, for each player \( i \), for every signal \( s \in S \), and for any alternative action \( a_i' \in A_i \), the following condition holds:
\begin{align*}
   \mathbb{E}_{\psi(\theta), \pi(s|\theta), \sigma}\!\Big[u_i(a_i,\! a_{-i}, \!\theta, \!s)\!\Big] \!\!\geq\!\! 
\mathbb{E}_{\psi(\theta), \pi(s|\theta), \sigma}\!\Big[u_i(a_i',\! a_{-i}, \!\theta, \!s)\!\Big]. 
\end{align*}
By the definition of \( a_i^* \), we have:
\[
\mathbb{E}_{\psi(\theta), \pi(s|\theta), \sigma}\!\Big[\!u_i(a_i^*,\! a_{-i}, \!\theta,\! s)\!\Big] \! > \!
\mathbb{E}_{\psi(\theta), \pi(s|\theta), \sigma}\!\Big[\!u_i(a_i,\! a_{-i},\! \theta,\! s)\!\Big]
\]
for all \( a_i \in A_i \setminus\{a_i^*\} \). Now, suppose for contradiction that there exists a BCCE \( \sigma \) where player \( i \) chooses an action \( a_i \neq a_i^* \) with positive probability for some signal \( s \). Then, consider a unilateral deviation by player \( i \) from the recommended action \( a_i \) to \( a_i^* \). Since \( a_i^* \) is a dominant action, \(  u_i(a_i^*, a_{-i}, \theta, s) > u_i(a_i, a_{-i}, \theta, s)\), for all \( (a_{-i}, \theta, s) \). Taking expectations over \( \psi(\theta) \) and \( \pi(s|\theta) \), we get:
\begin{align*}
   \mathbb{E}_{\psi(\theta), \pi(s|\theta)}\!\Big[u_i(a_i^*, a_{-i}, \theta, s)\!\Big] \!\!>\! 
\mathbb{E}_{\psi(\theta), \pi(s|\theta)}\!\Big[u_i(a_i, a_{-i}, \theta, s)\!\Big]. 
\end{align*}
This contradicts the BCCE condition, which requires that the expected utility from following the recommended action must be at least as good as any deviation. Thus, the only possible distribution \( \sigma \) that satisfies the BCCE condition is the degenerate distribution where:
\[
\sigma(a_1, \dots, a_n \mid s) = 1 \quad \text{if } a_i = a_i^* \text{ for all } i,
\]
and zero otherwise. Therefore, the BCCE is unique.
\end{proof}

\begin{remark}
 Since Lemma 2 imposes no assumptions on mechanisms, successful steering—which is defined as achieving zero directness gap between the players' action profiles and the target strategy point—is always feasible.
\end{remark}

Next, we analyze the ``investment'' game by categorizing it into two regions depending on the sign of \( z + y_B \). 

\begin{proposition}\label{Prop_1}
If $z + y_B > 0$, then $(I,I)$ is a strictly dominant strategy profile for both players, and therefore $(I,I)$ becomes the unique BCCE.
\end{proposition}
\begin{proof}
Since $y_G > 0$, we have \(  z + y_G > 0\). Thus, if $z + y_B > 0$, regardless of the state or the other player's action, the payoff for choosing action $I$ is strictly positive, while choosing action $N$ yields $0$. Thus, action $I$ strictly dominates action $N$ for each player in each state. As a result, due to Lemma \ref{lemma_dom}, $(I,I)$ becomes the unique  BCCE. 
\end{proof}

On the other hand, when \( z + y_B < 0 \), a strictly dominating strategy set does not exist, which necessitates the use of information and incentive design so that $(I,I)$ target point remains the unique BCCE point. To derive this from the definition of the BCCE, consider the inequality for player $i$ choosing $a_i = I$ over an alternative action $a_i' = N$:
\begin{align}
\sum_{\theta \in \{G,B\}} \sum_{s \in \{g,b\}} \sum_{a_{-i} \in \mathcal{A}_{-i}} &P(\theta,s,I,a_{-i}) \, u_i(I, a_{-i}, \theta) 
\geq \notag \\[-4pt]
\sum_{\theta \in \{G,B\}} \sum_{s \in \{g,b\}} \sum_{a_{-i} \in \mathcal{A}_{-i}} &P(\theta,s,N,a_{-i}) \, u_i(N, a_{-i}, \theta),
\end{align}
which can be simplified as:
\begin{align}\label{eqn:BCCE_cond}
\!\!\!\sum_{\theta \in \{G,B\}}\! \sum_{s \in \{g,b\}}\! \sum_{a_{-i} \in A_{-i}} P(\theta,s,I,a_{-i}) \, u_i(I, a_{-i}, \theta) \!\geq\! 0.\!
\end{align}

Using this definition, we can analyze how \((I,I)\) can be forced to be the unique BCCE point. Since the mediator's objectives are twofold—first, minimizing the amount paid to players, and second, reducing the directness gap to zero—we analyze three cases as follows: a) Steering with only information design, b) Steering with information design accompanied by sublinear payments, ensuring that the total amount paid remains finite, and  c) Steering with information design accompanied by linear payments. Thus, for these cases, we present the following two theorems demonstrating the non-feasibility of the first two cases, and then provide analysis of the last case.
\begin{theorem} \label{Lemma3}
There exists a game, within the provided setting, without strictly dominant strategies in which no-regret players cannot be steered toward a mediator's target strategy profile using only information design. Then, in such games, successful steering is not possible without incentives.
\end{theorem}
\begin{proof} 
First, we can rewrite (\ref{eqn:BCCE_cond}) as:
    \begin{align*}
    &\sum_{\theta \in \Theta} \sum_{s \in S } \!\sum_{a_{-i} \in \mathcal{A}_{-i}} \!\!\!\!\!\psi(\theta) \, \pi(s \! \mid\! \theta) \sigma(I,a_{-i} \!\mid\! s) u_i(I, a_{-i}, \theta)\!\! \geq\! 0
    \end{align*}
which gives:    
    \begin{align*}
    & \!\sum_{\theta \in \Theta}\! \psi(\theta) \! \Big(\! \pi(g \!\mid\! \theta) \!\big(\sigma(I,I\! \!\mid\! g)  u_i(I, I, \theta) \! \!+\! \!\sigma(I,N \!\mid\! g) u_i(I, N, \theta)\big)  \\ & \!\!+\!\pi(b \mid \theta)  \big(\sigma(I,I\! \mid \!b) \, u_i(I, I, \theta) \!\! +\!\! \sigma(I,N\! \mid \!b) \, u_i(I, N, \theta)\big)\! \Big) \! \!\geq 0. 
    \end{align*}
   For the sake of the counterexample, we  consider:  
    \begin{align*}
    \sigma(I \mid g) = \sigma(I \mid b) = \sigma(I,I \mid g) = \sigma(I,I \mid b)=1.
    \end{align*}
    Under this case, the BCCE no-deviation condition becomes:
    \begin{align*}
    \psi(B) (z + y_B ) + \psi(G) (z + y_G) \geq 0.
    \end{align*}
    Depending on $y_B$, $y_G$, and $z$, the inequality cannot be forced in all cases. Thus, in such games, information design is not enough and we need to provide incentives as well, which completes the proof. Also, this counterexample directly extends to the infeasibility of steering to a point in the BCE set with just information design, by assuming equal action probabilities and showing that the BCE set cannot be forced to be a singleton due to the inequality from the definition of the BCE set.
\end{proof}

\begin{theorem} \label{Lemma4}
There exists a game, within the provided setting, where information design, even when supported by sublinear payments, is insufficient to steer no-regret players toward the mediator's target strategy profile. Consequently, in such games, successful steering is not possible without constant average payments.
\end{theorem}

\begin{proof}
We denote the joint action probability of players as, \( \sigma(I,I|s) = \sigma^2(I|s) + \rho \sigma(I|s)(1-\sigma(I|s)) \), 
where \(\rho\) represents the correlation between the two no-regret players, due to inference between players' learning processes. Note that, introducing constant payments for the cases \((I, I)\) or \(I\) does not result in sub-linear average payments. Also, vanishing payments alone are insufficient, as they cannot enforce BCCE on the players indefinitely. Therefore, the only viable option is to introduce constant payments for the \((I, N)\) and $(N, I) $ cases. Formally from \eqref{eqn:BCCE_cond}, it can be seen that constant payments in these cases cannot ensure the players' convergence to a specific BCCE point, as:
    \begin{align*}
    &\sum_{\theta \in \Theta} \sum_{s \in S } \!\sum_{a_{-i} \in \mathcal{A}_{-i}} \!\!\!\!\!\psi(\theta) \, \pi(s \! \mid\! \theta) \sigma(I,a_{-i} \!\mid\! s) u_i(I, a_{-i}, \theta)\!\! \geq\! 0
    \end{align*}
which gives:    
    \begin{align*}
    & \!\sum_{\theta \in \Theta}\! \psi(\theta) \! \Big(\! \pi(g \!\mid\! \theta) \big(\sigma(I,I\! \!\mid\! g)  u_i(I, I, \theta) \! +\! \sigma(I,N \!\mid\! g) u_i(I, N, \theta)\big)  \\[-6pt] & \!\!+\!\pi(b \mid \theta)  \big(\sigma(I,I\! \mid \!b) \, u_i(I, I, \theta) \!\! +\!\! \sigma(I,N\! \mid \!b) \, u_i(I, N, \theta)\big) \Big) \! \!\geq 0. 
    \end{align*}
    Furthermore, this expression can be written as follows:    
\begin{align*}
    & \psi(B) \bigg[  \, \pi(g \mid B) \, [ (\sigma(I \mid g) - {\sigma(I \mid b)} ) (z+q) \\[-3pt] & + (\sigma(I,I \mid g) - {\sigma(I,I \mid b)})(y_B-q)] + \sigma(I \mid b) (z+q) \\[-3pt] & 
    + \sigma(I,I \mid b) (y_B-q) \bigg]  + \psi(G) \bigg[  \, \pi(g \mid G) \, [ (\sigma(I \mid g) \\[-3pt] & - \sigma(I \mid b) ) (z+q) + (\sigma(I,I \mid g) - \sigma(I,I \mid b))(y_G-q)]  \\[-3pt] & + \sigma(I \mid b) (z+q) + \sigma(I,I \mid b) (y_G-q) \bigg] \geq 0
\end{align*}
where $q$ is the constant payment introduced in $(N,I)$ and $(I,N)$. For the counterexample, we can assume the target point as, $\sigma(I \mid g) = \sigma(I \mid b)$, and $\sigma(I,I \mid g) = \sigma(I,I \mid b)$. For this specific case, we have that:
\begin{align*}
    & \psi(B)\bigg[(z+q) + [\sigma(I|b)+\rho(1-\sigma(I|b)](y_B-q)\bigg] \\ & + \psi(G)\bigg[(z+q)+(y_G-q)[\sigma(I|b)+\rho(1-\sigma(I|b))]\bigg] \geq 0 \\
    & \text{which can be further manipulated into:}  \\  
    & \bigg[\psi(B)y_B+\psi(G)y_G-q\bigg]\bigg[\sigma(I|b)(1-\rho)+\rho\bigg] \geq -z-q. 
\end{align*}
If $(\psi(B)y_B+\psi(G)y_G) < 0 $, we have: 
\begin{align*}
    & \sigma(I|b) \leq \bigg[ \frac{-z-q}{(\psi(B)y_B+\psi(G)y_G)-q}-\rho \bigg] \frac{1}{(1-\rho)}.
\end{align*}
Then, \(\sigma(I|b)\) happens to be the unique point in the BCCE set, only when both sides of inequality meets at \(0\). Thus, such an incentive scheme cannot guarantee successful steering to a specific point in the general case, concluding the proof. 
\end{proof}

As we have demonstrated that constant payments are necessary for successful steering in all games, we introduce a payment scheme by deriving a lower bound on the payments required for each round. For this purpose, we define the payment bound \( M \) as the total payments made, averaged over the sequences game has proceeded. Such a payment bound \( M \), which ensures successful steering, can be derived using the players' external regret. This follows from the fact that regret minimization ensures that players will match the best fixed hindsight action given the signal.

To identify the critical deviation, consider the scenario where a player deviates to action \( N \) in the bad state \( \theta = B \), and good state \( \theta = G \). As mediator commits to an incentive and information design scheme prior to observe the states of the world, we assume that it introduces constant payment across the all states.

Furthermore, toward bounding deviations, let \( D_\theta \) denote the total number of deviations for given state. By the definition of the overall-regret, the cumulative deviation cost must satisfy 
\begin{align}
    R_{ovr}(T) = D_G  (M+z+y_G) + D_B (M+z+y_B).
\end{align}
Then, to bound total deviations using the bound we have derived for $ R_{ovr}(T)$ we choose $M > z+y_B$, which makes both terms of the summand positive and $(I,I)$ dominant action profile. Then, denoting $\kappa = \min\{M+z+y_G, M+z+y_B\}$, we bound the directness gap with probability of at least $1-\delta$ as: 
\begin{align}\label{directness_gap_bound}
\delta(T) = \!\frac{D}{T} \!\le\! \frac{R_{ovr}(T)}{\kappa T}\!\le\!\frac{\sqrt{2}\;\tilde{R}\Big(T;\tfrac{\delta}{2}\Big)}{\kappa T}
\end{align}
\section{Stackelberg Optimization Framework}
\label{sec:stackelberg}
Given the infeasibility of introducing stationary signaling policy into the repeated game and the necessity of constant payments, we incorporate information design as a prior interaction between the mediator and the players. This interaction is modeled as a Stackelberg game, where the mediator is the leader, and the players are the followers who can coordinate. More specifically, since both players are symmetric and have identical preferences, they optimize over the same utility function. Furthermore, we provide solutions for the mediator's policy, where it prioritizes its steering objective. The model is based on the stationary game before the repeated game, with its payoff and action probability matrices given in Section II.

At the \textit{upper level}, the mediator selects the stationary signaling policy parameters \( (\alpha, \beta) \), which do not evolve over iterations, to maximize its own objective function. At the \textit{lower level}, given the mediator's choices \( (\alpha, \beta) \), each player chooses its strategy parameters \( (\alpha_g, \alpha_b, \gamma_g, \gamma_b) \) to maximize its expected utility \( \mathbb{E}[u] \) while satisfying the constraints imposed by the Bayesian Correlated Equilibrium (BCE).\footnote{Here, as the players have the symmetric utility functions, we denote the expected utility function of any of the two players by $\mathbb{E}[u]$.} For each follower's optimization problem, the expected utility to be maximized is given by:
\begin{equation}
    \mathbb{E}[u] = A_g \alpha_g + B_g \gamma_g + A_b \alpha_b + B_b \gamma_b,
\end{equation}
where the coefficients $A_g$, $B_g$, $A_b$, and $B_b$ are defined based on the upper-level variables \((\alpha, \beta)\) and other parameters:
\begin{align}
    &A_g \!= 
    \psi \alpha z + (1 - \psi) \beta z,  \\
    &B_g \!= 
    \psi \alpha y_G + (1 - \psi) \beta y_B,  \\
    &A_b \!=
    \psi (1 - \alpha) z + (1 - \psi) (1 - \beta) z, \\
    &B_b \!= 
    \!\psi (1 \!-\! \alpha) y_G \!+ \!(1 \!-\! \psi) (1\! -\! \beta) y_B,
\end{align}
where $A_j$'s are always positive by definition. Furthermore, for each type \(j \in \{g, b\}\), the decision variables must satisfy the following constraints, arising from the BCE and action probability constraints: 
\begin{align}
    &A_j \alpha_j \!+\! B_j \gamma_j\! \geq\! 0, \hspace{4pt}
    0 \!\leq \!\gamma_j \!\leq \!\alpha_j \!\leq 1, \hspace{4pt}
    1\! - \!2\alpha_j \!+ \!\gamma_j \!\geq \!0.
\end{align}

Note that the constraints for the good state (\(j = g\)) and the bad state (\(j = b\)) are independent. Hence, the lower-level problem can be decomposed into two separate two-dimensional linear programs. Each subproblem can be solved independently to obtain the optimal values of \((\alpha_j^*, \gamma_j^*)\) for \(j \in \{g, b\}\). For each \(j \in \{g, b\}\), the optimization subproblem is defined as:
\begin{align}
\max_{\{\alpha_j, \gamma_j\}}  \quad & \mathbb{E}[u_j] = A_{j} \alpha_j + B_{j} \gamma_j, \quad j\in \{g,b\}  \label{eq:objective_j}\\
\mbox{s.t.} \quad & A_j \alpha_j + B_j \gamma_j \geq 0 \label{eq:utility_constraint_j} \\
\quad & 0 \leq \gamma_j \leq \alpha_j \leq 1 \label{eq:bounds_j}\\
\quad & 1 - 2\alpha_j + \gamma_j \geq 0, \label{eq:additional_constraint_j}
\end{align}
where $\mathbb{E}[u] = \mathbb{E}[u_g]+\mathbb{E}[u_b]$. Each subproblem's feasible region is a convex polygon in the \((\alpha_j, \gamma_j)\) plane. Therefore, the optimal solution lies at a convex combination of the vertices. The potential vertices are \text{vertex A}: \((\alpha_j, \gamma_j) = (1, 1)\), \text{vertex B}: \((\alpha_j, \gamma_j) = \left(\frac{1}{2}, 0\right)\), and \text{vertex C}: \((\alpha_j, \gamma_j) = (0, 0)\), which has been plotted in Figure \ref{fig:feasible-region}. The optimal solution corresponds to the vertex with the highest value of \(\mathbb{E}[u_j]\). We compare the expected utilities $\mathbb{E}[u_j]$ evaluated at the points $A$, $B$, and $C$ to determine which vertex is optimal under different conditions. Vertex A is optimal if $\mathbb{E}[u_j(A)] > \mathbb{E}[u_j(B)]$ and $\mathbb{E}[u_j(A)] > \mathbb{E}[u_j(C)] $.
Substituting the expressions, we obtain $B_j > -\frac{A_j}{2}$ for vertex A to be optimal. %
Similarly, vertex B is optimal if $\mathbb{E}[u_j(B)] >\mathbb{E}[u_j(A)]$ and $\mathbb{E}[u_j(B)] > \mathbb{E}[u_j(C)]$
which imply that 

$ -\frac{A_j}{2} > B_j$. Finally, vertex C is optimal if $\mathbb{E}[u_j(C)] > \mathbb{E}[u_j(A)]$ and $\mathbb{E}[u_j(C)] > \mathbb{E}[u_j(B)]$.

Substituting the expressions, we obtain $A_j + B_j < 0$ and $\frac{A_j}{2}<0$ which is not possible as $A_j$ was defined to be positive. Hence, vertex C cannot be optimal. To cover all possible scenarios, we analyze boundary conditions where equalities hold. A convex combination of vertex A and vertex B is optimal if $\mathbb{E}[u_j(A)] = \mathbb{E}[u_j(B)]$ which implies that $B_{j} = -\frac{A_{j}}{2}$. In this case, any point in between vertices A and B is optimal. For this special case, we assume that both players prefer mediator preferred response which is $(\alpha_j^*, \!\gamma_j^*) \!=\! (1, \!1)$.
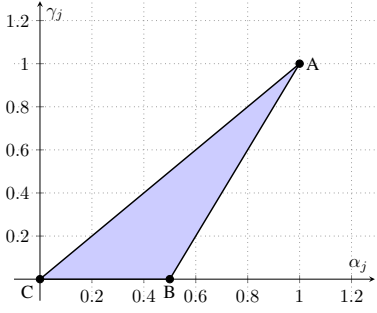
\begin{figure}[t]
    \centering
    \begin{tikzpicture}[scale=0.7]
        \begin{axis}[
            axis lines = middle,
            xlabel = {$\alpha_j$},
            ylabel = {$\gamma_j$},
            xmin = -0.1, xmax = 1.29,
            ymin = -0.1, ymax = 1.29,
            ticks = both,
            grid = major,
            grid style = {dotted, gray}
        ]
            \addplot[fill=blue!20, draw=none] coordinates {
                (0,0) (0.5,0) (1,1)
            } -- cycle;
    
            \addplot[domain=0:1, thick, samples=100] {x}; 
            \addplot[domain=0.5:1, thick, samples=100] {-1 + 2*x}; 
            \addplot[domain=0:0.5, thick, samples=100] {0}; 
            
            \addplot[only marks, mark=*] coordinates {(1,1)} node[right] {A};
            \addplot[only marks, mark=*] coordinates {(0.5,0)} node[below] {B};
            \addplot[only marks, mark=*] coordinates {(0,0)} node[below left] {C};
        \end{axis}
    \end{tikzpicture}
    \vspace{-0.2cm}
    \caption{Feasible region defined by vertices A, B, and C along with their convex combinations and constraints $0 \leq \gamma_j \leq \alpha_j$ and $1 - 2\alpha_j + \gamma_j \geq 0$, for the cases of $B_j > 0$ and $-B_j < A_j$.}
    \label{fig:feasible-region}
    \vspace{-0.4cm}
\end{figure} 

By combining all the cases above, we derive the closed-form solutions for the pair $(\alpha_j, \!\gamma_j)$ as follows:
\begin{equation}
    (\alpha_j^*, \gamma_j^*) =
    \begin{cases}
        (1, 1) & \text{if } B_{j} \geq -\frac{A_{j}}{2}, \\
        \left(\frac{1}{2}, 0\right) & \text{if } B_{j} < -\frac{A_{j}}{2}. 
    \end{cases}
    \label{eq:closed_form_solution_final}
\end{equation}

Here, we note that the players will always take action $(I,I)$ (corresponding to $(\alpha_j^*, \gamma_j^*) =(1,1)$ in (\ref{eq:closed_form_solution_final})) if $B_{j} \geq -\frac{A_{j}}{2}$ or take randomized actions $(I,N)$ or $(N,I)$ with probability $\frac{1}{2}$ if $B_{j} < -\frac{A_{j}}{2}$. As $A_{j}$ and $ B_{j}$ both depends on the mediator's policy $(\alpha, \beta)$, now we turn our attention to the mediator's optimization problem. 

In this case, we define the expected utility function of the mediator as: 
\begin{align}
\!\! \mathbb{E}[u_m] \!\!=\! \alpha_g  \Big(\!\psi\alpha\!\!+\!\!(1\!\!-\!\!\psi)\beta\!\Big)\!\!+\! \! \alpha_b \Big(\!\psi(1\!\!-\!\!\alpha)\!\!+\!\!(1\!-\!\psi)(1\!-\!\beta)\!\Big).\!\!  \label{eq:objective_m} 
\end{align}
Then, we solve the mediator's optimization problem by analyzing the following four cases: 

\begin{enumerate}
    \item \text{\!\!Case 1}: $\!(\alpha_g^*,\! \gamma_g^*) \!=\! \left(\frac{1}{2}, 0\right)$ and $\!(\alpha_b^*, \gamma_b^*) \! =\!$ $ \left(\frac{1}{2}, 0\right)$. Then, the utility of mediator becomes $\!\mathbb{E}[u_m]\! =\! 0.5$.
    
    \item \text{\!\!Case 2}:  $(\alpha_g^*, \gamma_g^*) \!\!=\!\!\left(\frac{1}{2}, 0\right)$ and $(\alpha_b^*, \gamma_b^*)\! =\! (1,1)$. Then, the expected utility becomes:
    \begin{align*}
         \!\!\!\!\!\!\!\!\mathbb{E}[u_m] \!\!= \!\!0.5 \Big(\psi\alpha\!+\!(1\!-\!\psi)\beta\Big)\!+\!  \Big(\psi(1\!-\!\alpha)\!+\!(1\!-\!\psi)(1\!-\!\beta)\Big)
    \end{align*}
    \normalsize
    \item \text{\!\!Case 3}: $\!(\alpha_g^*, \gamma_g^*) \!=\! (1,1)$ and $(\alpha_b^*, \gamma_b^*)$ $=\left( \frac{1}{2}, 0\right)$, respectively. Then, the expected utility becomes:
    \begin{align*}
         \!\!\!\!\!\!\!\!\mathbb{E}[u_m] \!\!= \!\! \Big(\psi\alpha\!+\!(1\!-\!\psi)\beta\Big)\!+\! 0.5 \Big(\psi(1\!-\!\alpha)\!+\!(1\!-\!\psi)(1\!-\!\beta)\Big)
    \end{align*}
    \normalsize
    \item \text{\!\!Case 4}: $(\alpha_g^*, \!\gamma_g^*)\! =\! (1,1)$ and $(\alpha_b^*, \gamma_b^*) \!=$ $ (1,1)$. Then, the utility of the mediator becomes $ \mathbb{E}[u_m] = 1$. 
\end{enumerate}

 Then, we can formally state the Stackelberg Equilibrium (SE) for the steering oriented mediator in the following theorem.
\begin{theorem}
The SE of the game is characterized by the mediator’s policy \((\alpha, \beta)\) and the players’ strategies \((\alpha_G, \gamma_G, \alpha_B, \gamma_B)\) for any \(0 \leq \eta \leq 1\) as follows:
\begin{align}
     (\alpha,& \beta, \alpha_G, \gamma_G, \alpha_B, \gamma_B ) =\!\nonumber \\& \!\!\!\!\!\!\!\!\!\!\begin{cases} 
      \biggl\{ \left(1, \frac{\psi(y_G+\frac{z}{2})}{(1-\psi)(-y_B-\frac{z}{2})}, 1, 1, 0.5, 0 \right), \\
     \left(\!0, 1 \!-\! \frac{\psi(y_G+\frac{z}{2})}{(1-\psi)(-y_B\!-\!\frac{z}{2})}, 0.5, 0, 1, 1 \!\right)\!\!\biggl\},& \!\!\!y_B\!<\! - \frac{\psi y_G+\frac{z}{2}}{1-\psi},\\
      (\eta, \eta, 1, 1, 1, 1),& \!\!\! y_B\!\geq\! - \frac{\psi y_G+\frac{z}{2}}{1-\psi}, 
   \end{cases}
\end{align}
for any \( 0 \! \leq \!\eta \!\leq \!1 \).

\end{theorem}
\begin{proof}
    It can be easily seen that the utility maximizer case for the mediator is Case 4. For this case to be optimal, we need $B_j \geq -\frac{A_j}{2}$ for $j\in\{g,b\}$ which imply that: 
    \begin{align*}
        \frac{(1-\psi)(-y_B-\frac{z}{2})}{\psi(y_G+\frac{z}{2})} \leq \min\Bigl\{\frac{\alpha}{\beta}, \frac{1-\alpha}{1-\beta} \Bigl\}.
    \end{align*}

    Choosing $\alpha=\beta$ on the right side  provides the loosest bounds on $y_B$ which is given by $ y_B\geq -\frac{\psi y_G+\frac{z}{2}}{1-\psi} $. Therefore, the SE happens to be \( 
    (\alpha, \beta, \alpha_G, \gamma_G, \alpha_B, \gamma_B ) \!=\! (\eta, \eta, 1, 1, 1, 1),
    \)
    for any \( 0 \! \leq \!\eta \!\leq \!1 \) when $ y_B\geq -\frac{\psi y_G+\frac{z}{2}}{1-\psi} $. 

     When Case 4 is not feasible, the mediator's next best options are Cases 2 and 3. For Case 2, due to mediator's objective, the optimal signaling probabilities are obtained with the lowest possible values for \((\alpha,\beta) \). Also, for this case to be optimal, we need $ B_g < -\frac{A_g}{2}$ and $B_b \geq -\frac{A_b}{2} $ which imply that:
    \begin{align}\label{eqn:Case2}
        \frac{\alpha}{\beta} < \frac{(1-\psi)(-y_B-\frac{z}{2})}{\psi(y_G+\frac{z}{2})} \leq \frac{1-\alpha}{1-\beta}.
    \end{align}

    Due to the mediator's objective function $\mathbb{E}[u_m]$, the utility of the mediator is maximized when
    \((\alpha,\beta) = \big(0, (1- \frac{\psi(y_G+\frac{z}{2})}{(1-\psi)(-y_B-\frac{z}{2})})\big)\), with the condition of $ y_B< -\frac{\psi y_G+\frac{z}{2}}{1-\psi} $. 
    For Case 3, the optimal signaling probabilities are obtained with the largest possible values of \((\alpha,\beta) \). Also, for Case 3 to be optimal, we need $B_g \geq -\frac{A_g}{2}$ and $B_b < -\frac{A_b}{2}$, implying that:
    \begin{align*}
       \frac{1-\alpha}{1-\beta} < \frac{(1-\psi)(-y_B-\frac{z}{2})}{\psi(y_G+\frac{z}{2})}\leq  \frac{\alpha}{\beta}. 
    \end{align*}

    In Case 3, the utility of the mediator is maximized when \((\alpha,\beta) = \big(1, \frac{\psi(y_G+\frac{z}{2})}{(1-\psi)(-y_B-\frac{z}{2})}\big)\), with the condition of $ y_B< -\frac{\psi y_G+\frac{z}{2}}{1-\psi} $. Since the mediator's utility is the same in both cases, we conclude that when $ y_B< -\frac{\psi y_G+\frac{z}{2}}{1-\psi} $ holds, there exist two Stackelberg Equilibria given by $ 
    (\alpha, \beta, \alpha_G, \gamma_G, \alpha_B, \gamma_B ) \in \Bigl\{ \Bigl(0,$ $ 1 - \frac{\psi(y_G+\frac{z}{2})}{(1-\psi)(-y_B-\frac{z}{2})}, 0.5, 0, 1, 1 \Bigl), \left(1, \frac{\psi(y_G+\frac{z}{2})}{(1-\psi)(-y_B-\frac{z}{2})}, 1, 1, 0.5, 0 \right)\Bigl\}$.    

    Finally, for Case 1, we need \( B_j < -\frac{A_j}{2} \) for \( j \in \{g,b\} \), which implies that:  
    \[
    \max\Bigl\{\frac{\alpha}{\beta}, \frac{1-\alpha}{1-\beta} \Bigr\} < \frac{(1-\psi)(-y_B - \frac{z}{2})}{\psi(y_G + \frac{z}{2})}.
    \]  
  Thus, choosing the loosest bound, we observe that it coincides with the intervals of Cases 2 and 3. Since, the mediator is better off in these cases, no SE arise from Case 1.
\end{proof}

Now, since the Stackelberg Equilibria are known, the mediator can guide players toward the one that is best for her. Thus, the general structure of the information-aided steering framework is outlined in Algorithm \ref{alg:algo1}. Given this structure, to bound the directness gap convergence rate, the Bayesian games can be reduced to a special adversarial bandit problem. In this setting, we specifically focus on the EXP3.P algorithm as it provides an ex-post regret bound with high probability rather than an expected regret bound, as ex-post regret is required for calculating the directness-gap convergence bound. Now, we first present the following lemma to streamline the analysis of high-probability regret bound for the EXP3.P class under non-uniform probability initialization. 

\begin{algorithm}[t] 
    \caption{Soft Inducement-Aided Steering Algorithm}
    \begin{algorithmic}[1]
        \renewcommand{\algorithmicrequire}{\textbf{Input:}}
        \renewcommand{\algorithmicensure}{\textbf{Output:}}
        
        \STATE \textbf{Input:} Initial game parameters $z$, $y_G$, $y_B$, $\psi$.
        \STATE  Mediator observes utility functions and inputs. Computes the Stackelberg equilibria. Then, commits to a signaling mechanism $\pi(\cdot|\theta)$ that leads to the best Stackelberg equilibrium for herself.
        \STATE Players observe the $\pi(\cdot|\theta)$, optimize their own utility functions, and arrive at the Stackelberg equilibrium.
        \FOR{each time step \( t = 1,2,\dots T\)}
            \STATE  Nature selects the state of the world \( \theta_t \), and mediator commits to a $\nu_i$.
            \STATE Mediator samples a public signal \( s_t \sim \pi(\cdot | \theta_t) \).
            \STATE Each player \( i \) observes \( s_t \) and selects action \( a_i^t \sim \sigma_{i,t}(\cdot | s_t, h_t) \).
            \STATE Players receive rewards based on joint actions and the realized state, \( v^{(t)}_i\), and update their decision strategies \( \sigma_{i,t}(\cdot | s_{t+1}, h_{t+1}) \) using no-regret learning algorithms.
        \ENDFOR
    \end{algorithmic}
    \label{alg:algo1}
\end{algorithm}

\begin{lemma}
\label{lem:estimator}
Fix $\beta\in(0,1]$. Then, for any $\delta\in(0,1)$, with probability of at least $1-\delta$,
\[
\sum_{t=1}^T g_{i,t}\;\le\;\sum_{t=1}^T \hat{g}_{i,t}\;+\;\frac{\ln(\delta^{-1})}{\beta}.
\]
\end{lemma}

\begin{proof}
The proof of Lemma~\ref{lem:estimator} follows directly from \cite[Lemma 3.1]{bubeck2012regret}.
\end{proof}

Consequently, we provide a high-probability regret bound for the EXP3.P class under non-uniform probability initialization.  

\begin{theorem}\label{thm:main}
Suppose that the parameters satisfy
\( 
\gamma \le \frac{1}{2},\) \(  (1+\beta)\,K\,\eta \le \gamma.
\)
Then, for any $\delta\in (0,1)$, with probability of at least $1-\delta$, the regret $R_T$ of Exp3.P, Algorithm \ref{alg:exp3p}, with initial weights $w_{i,1}=\pi_i$ is bounded by
\begin{align}\label{eq:hp-bound}
R_T \le \beta\,TK + \gamma\,T + (1+\beta)\,\eta\,K\,T - \frac{\ln(\pi^*)}{\eta} + \frac{\ln(K/\delta)}{\beta}
\end{align}
Choosing appropriate $\!\beta, \gamma, \eta$ one obtains the regret bound of: 
\begin{align}
R_T = \tilde{\mathcal{O}}\Bigl(\sqrt{TK}\big(4\sqrt{\ln(1/\pi^*)} +2\sqrt{\ln(K/\delta)}\big)\Bigr),
\end{align}
with probability of at least $1-\delta$, where \( \pi^* \) is the initial probability of choosing the best hindsight action.
\end{theorem}
\begin{proof}
First, from Algorithm \ref{alg:exp3p} it can be seen that, $p_{i,t}\ge \gamma/K$. Hence, we have
\begin{equation}
\label{eq:eta-hg-le-1}
\eta\,\widehat{g}_{i,t}\;\le\;\eta\frac{1+\beta}{p_{i,t}}
\;\le\; \frac{(1+\beta)\eta K}{\gamma} \;\le\; 1,
\end{equation}
which will be used with $e^x\le 1+x+x^2$ for $x\le 1$. Fix $k\in[K]$. Since
\(
\mathbb{E}_{i\sim p_t}\,[\widehat{g}_{i,t}] = g_{I_t,t} + \beta K,
\)
summing over $t$ yields
\begin{equation}
\label{eq:decomp}
\sum_{t=1}^T g_{k,t} - \sum_{t=1}^T g_{I_t,t}
= \beta K T \;+\; \sum_{t=1}^T g_{k,t} \;-\; \sum_{t=1}^T \mathbb{E}_{i\sim p_t}\,[\widehat{g}_{i,t}].
\end{equation}
Write $p_t=(1-\gamma)q_t+\gamma u$, where $u$ is uniform on $[K]$.
Using the exact identity
\(
-\mathbb{E}_q [X] = \frac{1}{\eta}\ln \mathbb{E}_q [e^{\eta(X-\mathbb{E}_q [X])}] - \frac{1}{\eta}\ln \mathbb{E}_q [e^{\eta X}],
\)
we have
\begin{align}
\!-\mathbb{E}_{i\sim p_t}\,[\widehat{g}_{i,t}]
\!\!&= -(1-\gamma)\,\mathbb{E}_{q_t}[\widehat{g}_{i,t} ]\;-\; \gamma\, \mathbb{E}_u [\widehat{g}_{i,t}] \nonumber\\
&= \!\!(1\!\!-\!\!\gamma)\!\Biggl[\!\frac{1}{\eta}\ln \mathbb{E}_{q_t}\! [e^{\eta(\widehat{g}_{i,t} -\mathbb{E}_{q_t}[\widehat{g}_{i,t})]}]
\!\!- \!\!\frac{1}{\eta}\!\ln \mathbb{E}_{q_t}\! [e^{\eta \widehat{g}_{i,t}}]\!\Biggr] \notag
\\ & \quad\;-\; \gamma\, \mathbb{E}_u[ \widehat{g}_{i,t}].
\label{eq:identity}
\end{align}
By \eqref{eq:eta-hg-le-1}, $x=\eta(\widehat{g}_{i,t}-\mathbb{E}_{q_t}\widehat{g}_{i,t})\!\le \!1$ almost surely, and thus we obtain
\begin{align*}
\!\ln \mathbb{E}_{q_t} [e^{\eta(\widehat{g}_{i,t}-\mathbb{E}_{q_t}[\widehat{g}_{i,t}])}]
&\!\!\;\le\; \!\!\!\mathbb{E}_{q_t}\!\big[\!e^{\eta(\widehat{g}_{i,t}-\mathbb{E}_{q_t}[\widehat{g}_{i,t}])}\!\!-\!\!1\!\!-\!\eta(\widehat{g}_{i,t}\!\!-\!\!\mathbb{E}_{q_t}[\widehat{g}_{i,t}])\!\big]\\&
\!\le \eta^2\,\mathbb{E}_{q_t}\!\big[(\widehat{g}_{i,t}-\mathbb{E}_{q_t}\widehat{g}_{i,t})^2\big]
\\&\!\le \eta^2\,\mathbb{E}_{q_t}\!\big[\widehat{g}_{i,t}^2\big].
\end{align*}
Also, $q_{i,t}\le p_{i,t}/(1-\gamma)$ and $\widehat{g}_{i,t}\le (1+\beta)/p_{i,t}$ imply
\begin{align*}
\mathbb{E}_{q_t}\!\big[\widehat{g}_{i,t}^2\big]
\!=\!\sum_{i=1}^K \!q_{i,t}\,\widehat{g}_{i,t}^2
 \le\! \frac{1}{1\!-\!\gamma}\!\sum_{i=1}^K p_{i,t}\,\widehat{g}_{i,t}\frac{1\!+\!\beta}{p_{i,t}}
\!= \!\frac{1\!+\!\beta}{1\!-\!\gamma}\!\sum_{i=1}^K \widehat{g}_{i,t}.
\end{align*}
Dropping the nonpositive term $-\gamma\,\mathbb{E}_u \widehat{g}_{i,t}\le 0$ from \eqref{eq:identity}, we have
\begin{equation}
\label{eq:one-step}
-\mathbb{E}_{i\sim p_t}\,[\widehat{g}_{i,t}
]\;\le\; (1+\beta)\eta \sum_{i=1}^K \widehat{g}_{i,t}
\;-\; \frac{1-\gamma}{\eta}\,\ln \mathbb{E}_{q_t} [e^{\eta \widehat{g}_{i,t}}].
\end{equation}
Since $w_{i,t+1}=w_{i,t}e^{\eta \widehat{g}_{i,t}}$ and $q_{i,t}=w_{i,t}/\sum_j w_{j,t}$, we have
\[
\mathbb{E}_{q_t} [e^{\eta \widehat{g}_{i,t}}] = \sum_{i=1}^K q_{i,t}\,e^{\eta \widehat{g}_{i,t}}
= \frac{\sum_{i=1}^K w_{i,t+1}}{\sum_{i=1}^K w_{i,t}}.
\]
Summing \eqref{eq:one-step} over $t=1,\dots,T$ and using
$\sum_{t=1}^T\sum_{i=1}^K \widehat{g}_{i,t}$ $ = \sum_{i=1}^K \widehat{G}_{i,T} \le K \max_j \widehat{G}_{j,T}$, we get
\begin{align}
\!-\!\!\sum_{t=1}^T \!\mathbb{E}_{i\sim p_t}\,[\widehat{g}_{i,t}]
\!&\le \!(1\!+\!\beta)\eta K \max_j \widehat{G}_{j,T}
\!- \!\frac{1\!-\!\gamma}{\eta}\,\!\!\sum_{t=1}^T \!\ln \!\frac{\sum_i \!w_{i,t+1}}{\sum_i \!w_{i,t}} \nonumber\\
&= \!(1\!+\!\beta)\eta K \max_j \widehat{G}_{j,T}
\!-\! \frac{1\!-\!\gamma}{\eta}\ln \frac{\sum_i \!w_{i,T+1}}{\sum_i \!w_{i,1}} \nonumber\\
&= \!(1\!+\!\beta)\eta K \max_j \widehat{G}_{j,T}
\!-\! \frac{1\!-\!\gamma}{\eta}\!\ln\! \!\Big[\!\!\sum_{i=1}^K\! \pi_i\,e^{\eta \widehat{G}_{i,T}}\!\Big]
\label{eq:summed}
\end{align}
because $\sum_i\! w_{i,1}\!=\!\sum_i\! \pi_i\!=\!1$ and $w_{i,T+1}\!=\!\pi_i e^{\eta \widehat{G}_{i,T}}$. Let $i^\star\in$ $\!\arg\max_i \sum_{t=1}^T g_{i,t}$; then
\(
\sum_i \pi_i e^{\eta \widehat{G}_{i,T}}
\!\ge\! \pi^\star e^{\eta \widehat{G}_{i^\star,T}}
\).
Thus,
\begin{align}
\label{eq:summed2}
-\!\!\sum_{t=1}^T \! \mathbb{E}_{i\sim p_t}\,[\widehat{g}_{i,t}]
\!\le& (1\!+\!\beta)\eta K \max_j \widehat{G}_{j,T}
-  (1-\gamma)\widehat{G}_{i^\star,T}\notag \\ &- \frac{1\!-\!\gamma}{\eta}\ln(\pi^\star)
\end{align}
Using Lemma \ref{lem:estimator}, we have
\(
\max_j \widehat{G}_{j,T} \ge \max_j \sum_{t=1}^T g_{j,t} - \frac{\ln(K/\delta)}{\beta}
= \sum_{t=1}^T g_{i^\star,t} - \frac{\ln(K/\delta)}{\beta}
\).
Plugging this lower bound into \eqref{eq:summed2} gives
\begin{align*}
-\!\sum_{t=1}^T \!\mathbb{E}_{i\sim p_t}\,[\widehat{g}_{i,t}]
\!\le \!& \!-\!\Bigl[1\!-\!\gamma\!-\!(1\!+\!\beta)\eta K\Bigr]\!\sum_{t=1}^T g_{i^\star,t} \!- \!\frac{1\!-\!\gamma}{\eta}\ln(\pi^\star)
\\ & + \Bigl[1-\gamma-(1+\beta)\eta K\Bigr]\frac{\ln(K/\delta)}{\beta}.
\end{align*}
Combining with \eqref{eq:decomp} for $k=i^\star$ yields
\begin{align*}
R_T
\le\; &   \beta K T
\;+\;\Bigl(\gamma + (1+\beta)\eta K \Bigr)\sum_{t=1}^T g_{i^\star,t}
- \frac{1-\gamma}{\eta}\ln(\pi^\star)
\\& + \Bigl[1-\gamma-(1+\beta)\eta K\Bigr]\frac{\ln(K/\delta)}{\beta}.
\end{align*}
Using $\sum_{t=1}^T g_{i^\star,t}\le T$ and $1-\gamma-(1+\beta)\eta K\le 1$, we obtain
\[
R_T
\le \beta K T + \gamma T + (1+\beta)\eta K T
- \frac{1-\gamma}{\eta}\ln(\pi^\star)
+ \frac{\ln(K/\delta)}{\beta}.
\]
Finally, since $1-\gamma\le 1$ and $\ln(\pi^\star)\le 0$, loosening $-\frac{1-\gamma}{\eta}\ln(\pi^\star)\le -\frac{1}{\eta}\ln(\pi^\star)$ gives \eqref{eq:hp-bound}. Finally, we set \( 
\beta=\sqrt{\frac{\ln(K/\delta)}{K T}},
\eta=\sqrt{\frac{\ln(1/\pi^\star)}{K T}},
\gamma=(1+\beta)\,K\,\eta.
\) 
Plugging into \eqref{eq:hp-bound} and grouping terms gives
\begin{align*}
R_T
\!\le &\!
\underbrace{\sqrt{\!K T \ln(K/\delta)}}_{\beta K T}
\!+\!
\underbrace{(1\!+\!\beta)\sqrt{\!K T \ln(1/\pi^\star)}}_{\gamma T}\!+\!
\underbrace{\sqrt{\!K T \ln(1/\pi^\star)}}_{-\ln(\pi^\star)/\eta}
\\ &\! +
\underbrace{(1+\beta)\sqrt{K T \ln(1/\pi^\star)}}_{(1+\beta)\eta K T}
+
\underbrace{\sqrt{K T \ln(K/\delta)}}_{\ln(K/\delta)/\beta},
\end{align*}
which simplifies to the announced
\(
\tilde{\mathcal{O}}\!\big(\sqrt{K T}(4\sqrt{\ln(1/\pi^\star)}+2\sqrt{\ln(K/\delta)})\big).
\)
\end{proof}

Then, the result for directness gap convergence rate can be stated as follows:
\begin{theorem}
Let \(R(T) \) denote the regret of regular no-regret learners and \( R^*(T) \) denote the regret of the SE-initiated players, for a given signal type. Then, we have 
\( 
R^*(T) = \tilde{\mathcal{O}}\left(\sqrt{TK}\left(4\sqrt{\ln(1/\pi^*)} + 2\sqrt{\ln(K/\delta)}\right)\right),
\)
which improves upon 
\( 
R(T) = \tilde{\mathcal{O}}\left(\sqrt{TK}\left(4\sqrt{\ln(K)} + 2\sqrt{\ln(K/\delta)}\right)\right).
\)
This result extends to the directness gap convergence rate as
\[
\delta^*(T) \!=\! \frac{R_{ovr}^*(T)}{\kappa T} \!= \!\tilde{\mathcal{O}}\left( \frac{2\sqrt{2K\ln(2K/\delta)}\!+\!4\sqrt{\gamma K}}{\sqrt{T}\kappa}\right)
\]
which provides a tighter bound compared to
\[
\!\delta(T) \! =\! \frac{R_{ovr}(T)}{\kappa T} \!= \!\tilde{\mathcal{O}}\!\left(\! \frac{2\sqrt{2K\ln(2K/\delta)}\!+\!4\sqrt{2K\ln K}}{\sqrt{T}\kappa} \! \right)
\]

Here, $\gamma < 1$ denotes the fraction of timesteps—multiplied by $\ln\left(\frac{1}{\pi^*}\right)$—in which the received signal did \emph{not} lead the players to $\alpha_s = 1$ in the Stackelberg equilibrium.

\end{theorem}
\begin{proof}
    First part of the theorem trivially follows from Theorem 5 and $\pi^* \geq \frac{1}{2} =\frac{1}{K}$. To derive the upper bound for $R^*_{ovr}$, we consider the worst SE arising from Case 3. Let $T_1$ denote the good signal rounds, $T_2$ denote the bad signal rounds, and $T=T_1+T_2$. Then, having $K=2=\frac{1}{\pi^*_{T_2}}$: 
    \begin{align*}
    &\! (2\sqrt{K}\!\sqrt{\ln(2K/\delta)}(\!\sqrt{T_1}\!+\!\sqrt{T_2}\!) +  \!4\sqrt{T_2K}\!\sqrt{\ln2})\! \notag \\ \notag
    & \leq (2\sqrt{2(T_1+T_2)K}\sqrt{\ln(2K/\delta)}+ 4\sqrt{T_2K}\sqrt{\ln2}) = f(T)\\ 
    & < (2\sqrt{2TK}\sqrt{\ln(2K/\delta)}+ 4\sqrt{2TK}\sqrt{\ln2}) = g(T) 
    \end{align*}
    where the $\tilde{\mathcal{O}}(f(T)) = R^*_{ovr}(T)$, and $\tilde{\mathcal{O}}(g(T)) = R_{ovr}(T)$. Finally, the second part follows from the above inequalities, along with \eqref{directness_gap_bound}, which concludes the proof.
\end{proof}

\section{Numerical Experiments}
In the experiments, we compare the directness gap of regular players with that of SE-initiated players. We analyze two different cases. In the first case, 
\( 
y_B \geq - \frac{\psi y_G + \frac{z}{2}}{1 - \psi},
\) 
the SE is initiated as
\( 
(\alpha, \beta, \alpha_G, \gamma_G, \alpha_B, \gamma_B) = (\eta, \eta, 1, 1, 1, 1),
\) 
and the convergence results are depicted in Figure~\ref{fig:exp1}. In the second case,
\( 
y_B < - \frac{\psi y_G + \frac{z}{2}}{1 - \psi},
\) 
the SE is initiated as
\( 
(\alpha, \beta, \alpha_G, \gamma_G, \alpha_B, \gamma_B) = (0, 0, 0.5, 0, 1, 1),
\) 
for the chosen parameter values, with the results shown in Figure~\ref{fig:fig2}.

\begin{figure} 
    \centering
    \includegraphics[scale=0.35]{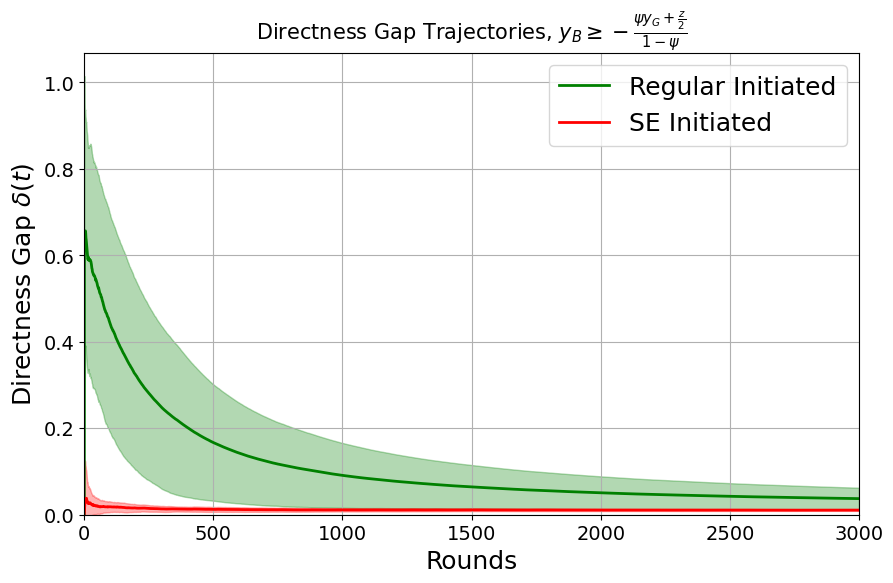}
    \vspace{-0.25cm}
    \caption{Trajectories for the \(y_B\geq - \frac{\psi y_G+\frac{z}{2}}{1-\psi}\) case. With, \(\psi = 0.7\), \(\alpha = 0.7\), \(\beta = 0.7\), \(z = 0.2\), \(y_G = 1\), \(y_B = -0.05\), penalty parameter \(M = 0.24\), the learning rate \(\eta = 0.05\) for the EXP3.P algorithm, .
    }
    \label{fig:exp1}
    \vspace{-0.4cm}
\end{figure}

Each player maintains two separate instances of the EXP3.P algorithm, one for each signal instance. As described in the theoretical analysis of the previous sections, in our experiments, after the mediator commits to a stationary incentive and information-signaling scheme following the initial SE, it maintains this scheme throughout all iterations without updating it.

In Figures~\ref{fig:exp1} and~\ref{fig:fig2}, we report on the evolution of the directness gap, where solid lines denote the mean of 50 independent runs and the shaded area represents one standard deviation around the mean. To preserve exploration, probabilities on the order of \(10^{-2}\) were used in place of zero probabilities that would otherwise be assigned by the SEs at initialization.

Finally, we observe in the numerical experiments that our soft-inducement-assisted prior Stackelberg game framework outperforms the randomly initialized no-regret players, as also proven theoretically.

\section{Conclusion}

In this work, we have addressed the problem of steering no-regret players with incentive and information design. We have shown that successful steering is not feasible through information design alone, or accompanied with sublinear payments. First, we derived a lower bound on the required average payments. Then, we proposed an information design-based initiation of players for the repeated normal-form game. Next, we established improved directness gap convergence rate for the proposed framework. Finally, we supported these improved bounds with numerical experiments. Future work will focus on improving asymptotic regret bounds through carefully crafted information design in games with side information and no-regret players.

\bibliographystyle{IEEEtran}
\bibliography{references}

\begin{thebibliography}{10}
\providecommand{\url}[1]{#1}
\csname url@samestyle\endcsname
\providecommand{\newblock}{\relax}
\providecommand{\bibinfo}[2]{#2}
\providecommand{\BIBentrySTDinterwordspacing}{\spaceskip=0pt\relax}
\providecommand{\BIBentryALTinterwordstretchfactor}{4}
\providecommand{\BIBentryALTinterwordspacing}{\spaceskip=\fontdimen2\font plus
\BIBentryALTinterwordstretchfactor\fontdimen3\font minus \fontdimen4\font\relax}
\providecommand{\BIBforeignlanguage}[2]{{%
\expandafter\ifx\csname l@#1\endcsname\relax
\typeout{** WARNING: IEEEtran.bst: No hyphenation pattern has been}%
\typeout{** loaded for the language `#1'. Using the pattern for}%
\typeout{** the default language instead.}%
\else
\language=\csname l@#1\endcsname
\fi
#2}}
\providecommand{\BIBdecl}{\relax}
\BIBdecl

\bibitem{bacsar2024inducement}
T.~Ba{\c{s}}ar, ``Inducement of desired behavior via soft policies,'' \emph{International Game Theory Review}, vol.~26, no.~02, p. 2440002, 2024.

\bibitem{babichenko2021multi}
Y.~Babichenko, I.~Talgam-Cohen, H.~Xu, and K.~Zabarnyi, ``Multi-channel {B}ayesian persuasion,'' \emph{arXiv preprint arXiv:2111.09789}, 2021.

\bibitem{mccloskey1995one}
D.~McCloskey and A.~Klamer, ``One quarter of {G}{D}{P} is persuasion,'' \emph{The American Economic Review}, vol.~85, no.~2, pp. 191--195, 1995.

\bibitem{egorov2020persuasion}
G.~Egorov and K.~Sonin, ``Persuasion on networks,'' National Bureau of Economic Research, Tech. Rep., 2020.

\bibitem{johnson2006simple}
J.~P. Johnson and D.~P. Myatt, ``On the simple economics of advertising, marketing, and product design,'' \emph{American Economic Review}, vol.~96, no.~3, pp. 756--784, 2006.

\bibitem{ke2022information}
T.~T. Ke, S.~Lin, and M.~Y. Lu, \emph{Information {D}esign of {O}nline {P}latforms}.\hskip 1em plus 0.5em minus 0.4em\relax SSRN, 2022.

\bibitem{goldstein2016bayesian}
I.~Goldstein and C.~Huang, ``Bayesian persuasion in coordination games,'' \emph{American Economic Review}, vol. 106, no.~5, pp. 592--596, 2016.

\bibitem{goldstein2018stress}
I.~Goldstein and Y.~Leitner, ``Stress tests and information disclosure,'' \emph{Journal of Economic Theory}, vol. 177, pp. 34--69, 2018.

\bibitem{bosmans2022systematic}
C.~Yzermans, ``A systematic review of rapid needs assessments and their usefulness for disaster decision making: methods, strengths and weaknesses and value for disaster relief policy,'' \emph{International Journal of Disaster Risk Reduction}, vol.~71, p. 102807, 2022.

\bibitem{blinder2015financial}
A.~S. Blinder and M.~Zandi, ``The financial crisis: Lessons for the next one,'' \emph{Center on Budget and Policy Priorities: Policy Futures}, 2015.

\bibitem{harsanyi1967}
J.~C. Harsanyi, ``Games with incomplete information played by “{B}ayesian” {P}layers, {I–III} {Part I.} the basic model,'' \emph{Management Science}, vol.~14, no.~3, 1967.

\bibitem{bubeck2012regret}
S.~Bubeck, N.~Cesa-Bianchi \emph{et~al.}, ``Regret analysis of stochastic and nonstochastic multi-armed bandit problems,'' \emph{Foundations and Trends in Machine Learning}, vol.~5, no.~1, pp. 1--122, 2012.

\bibitem{von1934marktform}
H.~von Stackelberg, \emph{Marktform und Gleichgewicht}.\hskip 1em plus 0.5em minus 0.4em\relax J. Springer, 1934.

\bibitem{mguni2019coordinating}
D.~Mguni, J.~Jennings, S.~V. Macua, E.~Sison, S.~Ceppi, and E.~M. De~Cote, ``Coordinating the crowd: Inducing desirable equilibria in non-cooperative systems,'' \emph{arXiv preprint arXiv:1901.10923}, 2019.

\bibitem{liu2022inducing}
B.~Liu, J.~Li, Z.~Yang, H.-T. Wai, M.~Hong, Y.~Nie, and Z.~Wang, ``Inducing equilibria via incentives: Simultaneous design-and-play ensures global convergence,'' \emph{Advances in Neural Information Processing Systems}, vol.~35, pp. 29\,001--29\,013, 2022.

\bibitem{zhang2023steering}
B.~H. Zhang, G.~Farina, I.~Anagnostides, F.~Cacciamani, S.~M. McAleer, A.~A. Haupt, A.~Celli, N.~Gatti, V.~Conitzer, and T.~Sandholm, ``Steering no-regret learners to a desired equilibrium,'' \emph{arXiv preprint arXiv:2306.05221}, 2023.

\bibitem{crawford1982strategic}
V.~P. Crawford and J.~Sobel, ``Strategic information transmission,'' \emph{Econometrica: Journal of the Econometric Society}, pp. 1431--1451, 1982.

\bibitem{kamenica2011bayesian}
E.~Kamenica and M.~Gentzkow, ``Bayesian persuasion,'' \emph{American Economic Review}, vol. 101, no.~6, pp. 2590--2615, 2011.

\bibitem{bergemann2016bayes}
D.~Bergemann and S.~Morris, ``Bayes correlated equilibrium and the comparison of information structures in games,'' \emph{Theoretical Economics}, vol.~11, no.~2, pp. 487--522, 2016.

\bibitem{hartline2015no}
J.~Hartline, V.~Syrgkanis, and E.~Tardos, ``No-regret learning in {B}ayesian games,'' \emph{Advances in Neural Information Processing Systems}, 2015.

\bibitem{cesa2006prediction}
N.~Cesa-Bianchi and G.~Lugosi, \emph{Prediction, {L}earning, and {G}ames}.\hskip 1em plus 0.5em minus 0.4em\relax Cambridge University Press, 2006.

\bibitem{shao2003mathematical}
J.~Shao, \emph{Mathematical Statistics}.\hskip 1em plus 0.5em minus 0.4em\relax Springer New York, NY, 2003.

\end{thebibliography}
\begin{figure} 
    \centering 
    \includegraphics[scale=0.35]{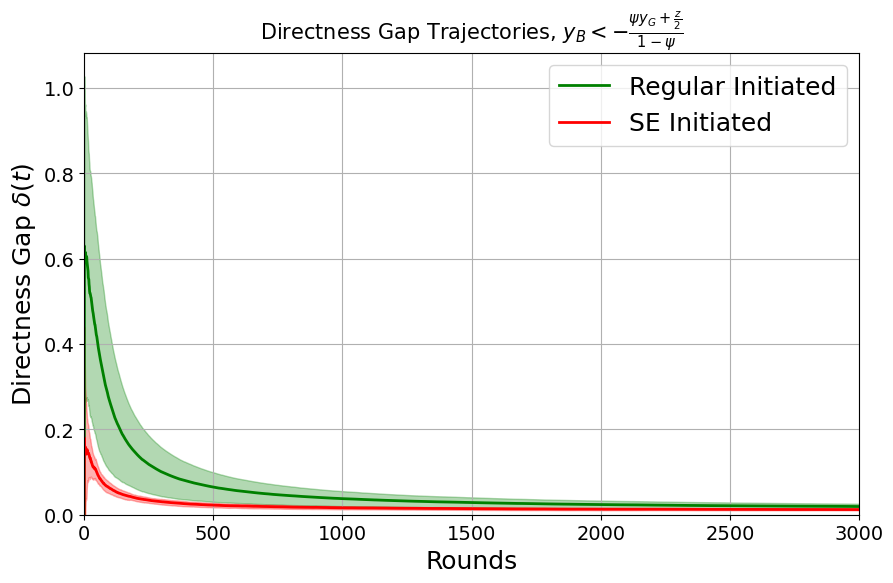}
    \caption{Trajectories for the \(y_B< - \frac{\psi y_G+\frac{z}{2}}{1-\psi}\) case. With, \(\psi = 0.7\), \(\alpha = 0\), \(\beta = 0\), \(z = 0.2\), \(y_G = 0.1\), and \(y_B = -0.56\), penalty parameter \(M = 0.60\), the learning rate \(\eta = 0.05\) for the EXP3.P algorithm.}
    \label{fig:fig2}
\end{figure}
\end{document}